\documentclass[10pt, letterpaper, twocolumn, conference]{IEEEtran}

\usepackage{graphicx}
\usepackage{subfigure}
\usepackage{amsmath}
\usepackage{amssymb}
\usepackage{cite}
\usepackage{cuted}
\usepackage{stfloats}
\usepackage{color}
\usepackage{flushend}
\usepackage{epstopdf}
\IEEEoverridecommandlockouts
\newtheorem{theorem}{\bf Theorem}
\newtheorem{lemma}{\bf Lemma}

\newcommand{\cost}{\overline{\mathcal{J}}}
\newcommand\independent{\protect\mathpalette{\protect\independenT}{\perp}}
\def\independenT#1#2{\mathrel{\rlap{$#1#2$}\mkern2mu{#1#2}}}

\newcommand{\indi}[1]{{1\hspace{-2.3mm}{1}}_{\left\{#1\right\}}}
\newcommand{\m}[1]{\mathbf{#1}^m}
\newcommand{\lo}[1]{\log_2\left(#1\right)}
\newcommand{\lon}[1]{\ln\left(#1\right)}
\newcommand{\mk}[2]{\mathbf{#1}^{#2}}

\newcommand{\co}[1]{\ensuremath{\underline{\underline{\text{C}_{#1}}}}}

\newcommand{\expect}[1]{\mathbb{E}\left[{#1}\right]}
\newcommand{\expectp}[2]{\mathbb{E}_{#1}\left[{#2}\right]}

\definecolor{DarkGreen2}{rgb}{0.00,0.6,0.08}

\title{Implicit and explicit communication in decentralized control}
\author{Pulkit Grover and Anant Sahai\\ Department of EECS, University of California at Berkeley, CA-94720, USA\\ \{pulkit, sahai\}@eecs.berkeley.edu}

\begin{document}\maketitle
\begin{abstract}
There has been substantial progress recently in understanding toy problems of purely \textit{implicit} signaling. These are problems where the source and the channel are implicit --- the message is generated endogenously by the system, and the plant itself is used as a channel. In this paper, we explore how implicit and explicit communication can be used synergistically to reduce control costs. 

The setting is an extension of Witsenhausen's counterexample where a rate-limited external channel connects the two controllers. Using a semi-deterministic version of the problem, we arrive at a binning-based strategy that can outperform the best known strategies by an arbitrarily large factor. 

We also show that our binning-based strategy attains within a constant factor of the optimal cost for an asymptotically infinite-length version of the problem uniformly over all problem parameters and all rates on the external channel. For the scalar case, although our results yield approximate optimality for each fixed rate, we are unable to prove approximately-optimality uniformly over all rates. 
\end{abstract}
\section{Introduction}
\label{sec:intro}
In his layered approach to design of decentralized control systems~\cite{VaraiyaLayered}, Varaiya dedicates an entire layer for coordinating the actions of various agents. The question is: how can the agents build this coordination?

The most natural way to build coordination is through communication. To begin with, let us assume that the source and the channel have been specified explicitly. Even with this simplification, the general problem of multiterminal information theory has proven to be hard. The community therefore resorted to building a bottom-up theory that starts from Shannon's toy problem of point-to-point communication~\cite{ShannonOriginalPaper}. The insights and tools obtained from this toy problem have helped immensely in the continuing development of multiterminal information theory. 

A more accurate model of a dynamic control system is where the source can \textit{evolve} with time, reflecting the impact of random perturbations and control actions. A counterpart of Shannon's point-to-point toy problem that models evolution due to random perturbations is a problem of communicating an unstable Markov source across a channel. The problem is reasonably well understood~\cite{WongBrockett,BorkarMitter,TatikondaThesis,SahaiThesis,Mateev}, and again,  building on the understanding for this toy problem, the community has begun exploring multicontroller problems~\cite{YukselBasar2,YukselBasarMulticontroller}.

Do the above models encompass the possible ways of building coordination? Because these models are motivated by an architectural separation of estimation and control, they do not model the impact of control actions in state evolution\footnote{Communication has also been  used to build coordination by generating correlation between random variables~\cite{Cuff08,CuffCoordination}.}. Is this aspect important? Indeed, in decentralized control systems, it is often possible to modify what is to be communicated before communicating it. But at times, it is also often unclear what medium to use for communicating the message~\cite[Ch. 1]{PulkitThesis}. That is, the sources and the channels may not be not as explicit as assumed in traditional communication models. To understand this issue, we informally define \textit{implicit communication} to be one of the following two phenomena arising in decentralized control:
\begin{itemize}
\item Implicit message: the message itself is generated endogenously by the control system.
\item Implicit channel: the system to be controlled is used as a channel to communicate.
\end{itemize}
The first phenomenon, that of implicit messages, poses an intellectual challenge to information theorists. How does one communicate a message that is endogenously generated, and hence can potentially be affected by the policy choice? 

The second phenomenon, that of viewing the plant as an implicit communication channel, is challenging from a control theoretic standpoint. The control actions now perform a dual role --- control of the system (\textit{i.e.} minimizing immediate costs), and communication through the system (presumably to lower future costs).

\begin{figure}[htbp] 
   \centering
   \includegraphics[width=2in]{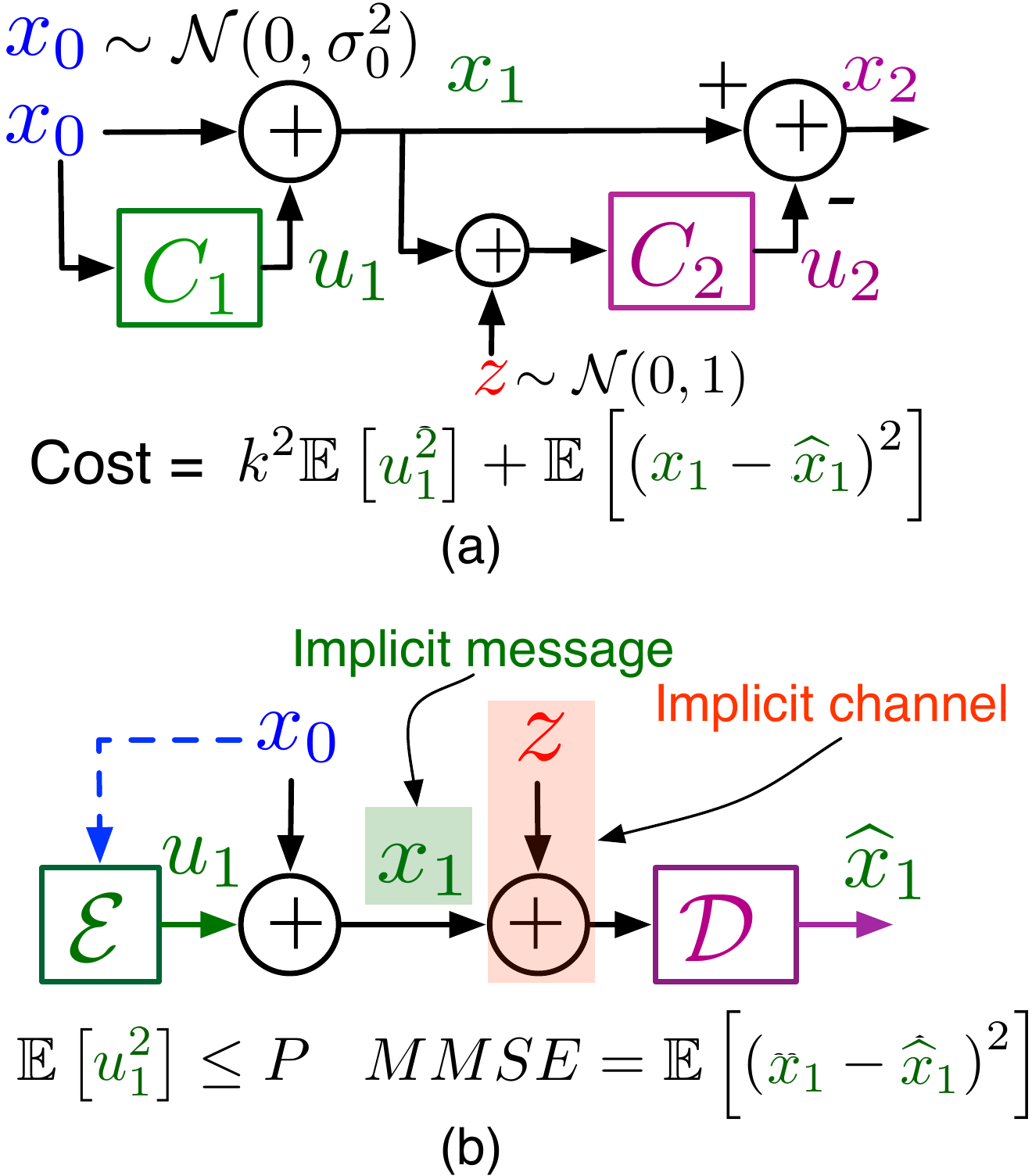} 
   \caption{The Witsenhausen counterexample, shown in (a) is the minimalist toy problem that exhibits the two notions of implicit communication, shown in (b), which is an equivalent representation~\cite{WitsenhausenJournal}. }
   \label{fig:wcounter}
\end{figure}

The counterpart of Shannon's point-to-point problem in implicit communication is a decentralized two-controller problem called Witsenhausen's counterexample~\cite{Witsenhausen68} shown in Fig.~\ref{fig:wcounter}. The message, state $x_1$, is implicit, because it can be affected by the input $u_1$ of the first controller. The channel is implicit because the system state itself is used to communicate the message.

Despite substantial efforts of the community, the counterexample remains unsolved, and due to this the community could not build on the problem to address larger control networks of this nature. Recently, however, we showed that using the input to quantize the state (complemented by linear strategies) attains within a constant factor of the optimal cost uniformly over all problem parameters for the counterexample and its vector extensions~\cite{WitsenhausenJournal,FiniteLengthsWitsenhausen}. Building on this provable approximate-optimality we have been able to obtain similar results for many extensions to the counterexample\footnote{Approximate-optimality results of this nature have proven useful in information theory as well ---  building on smaller problems~\cite{EtkinOneBit}, significant understanding has been gained about larger systems~\cite{SalmanThesis}. }~\cite{ITW10Paper,Allerton09Paper,ISIT10Witsenhausen,CDC10Paper,PulkitThesis}.

When is it useful to communicate implicitly? To understand this, Ho and Chang~\cite{HoChang} introduce the concept of partially-nested information structures. Their results can be interpreted in the following manner: when transmission delay across a noiseless, infinite-capacity external channel is smaller than the propagation delay of implicit communication, there is no advantage in communicating implicitly\footnote{The same conclusion is drawn in work of Rotkowitz an Lall~\cite{RotkowitzLall} (as an application of quadratic-invariance) and that of Y\"{u}ksel~\cite{YukselStochastic} in more general frameworks.}. The system designer always has the engineering freedom to attach an external channel. Can this external channel obviate the need to consider implicit communication?

In practice, however, the channel is \textit{never} perfect. In~\cite[Ch. 1]{PulkitThesis}, we compare problems of implicit and explicit communication where the respective channels are noisy. Assuming that the weights on quadratic costs on inputs and reconstruction are the same for implicit and explicit communication, we show that implicit communication can outperform various architectures of explicit communication by an arbitrarily large factor! The gain is due to implicit nature of the messages --- the simplified source after actions of the controller can be communicated with much greater fidelity for the same power cost.

So an external channel should not be thought of as a substitute for implicit communication. But if an external channel is available, how should it be used in conjunction with implicit communication? To examine this, we consider an extension of Witsenhausen's counterexample (shown in Fig.~\ref{fig:synergy}) where an external channel connects the two controllers. A special case when the channel is power constrained and has additive Gaussian noise has been considered by Shoarinejad \textit{et al}~\cite{Shoarinejad} and Martins~\cite{MartinsSideInfo}. Shoarinejad \textit{et al} observe that when the channel noise variance diverges to infinity, the problem approaches Witsenhausen's counterexample, while linear strategies are optimal in the limit of zero noise. Martins considers the case of finite noise variance and shows that in some cases, there exist nonlinear strategies that outperform all linear strategies. 

In Section~\ref{sec:deterministic}, we provide an improvement over Martins's strategy based on intuition obtained from a semi-deterministic version of the problem. In Section~\ref{sec:gaussian}, we show that our strategy can outperform Martins's strategy by an arbitrarily large factor. Because we interpret the problem as communication across two parallel channels --- an implicit one and an explicit one --- our strategy ensures that the information on implicit and explicit channels is essentially orthogonal. Without the implicit channel output, the message our strategy sends on the explicit channel would yield little information about the state. But the observations on the two channels jointly reveal a lot more about the state. This eliminates a redundancy in Martins's strategies where the same message is duplicated over the implicit and explicit channels. In this sense, our results here also provide a justification for the utility of the concept of implicit communication.

For simplicity, we assume a fixed-rate noiseless external channel for most of the paper. In Section~\ref{sec:asymptotic}, our binning strategy is proved to be approximately optimal for all problem parameters and all rates on the external channel for an \textit{asymptotic vector version of the problem}. In Section~\ref{sec:scalar}, using tools from large-deviation theory and KL-divergence, we obtain a lower bound on the costs for finite vector-lengths. Using this lower bound, we show that our improved strategy is within a constant factor of optimal for any fixed rate $R_{ex}$ on the external channel for the \textit{scalar case}. However, we do not yet have an approximately-optimal solution that is uniform over \textit{external channel's rate} ---  the ratio of upper and lower bounds diverges to infinity as $R_{ex}\rightarrow\infty$. We conclude in Section~\ref{sec:conclusions}.

\begin{figure}[htbp] 
   \centering
   \includegraphics[width=2in]{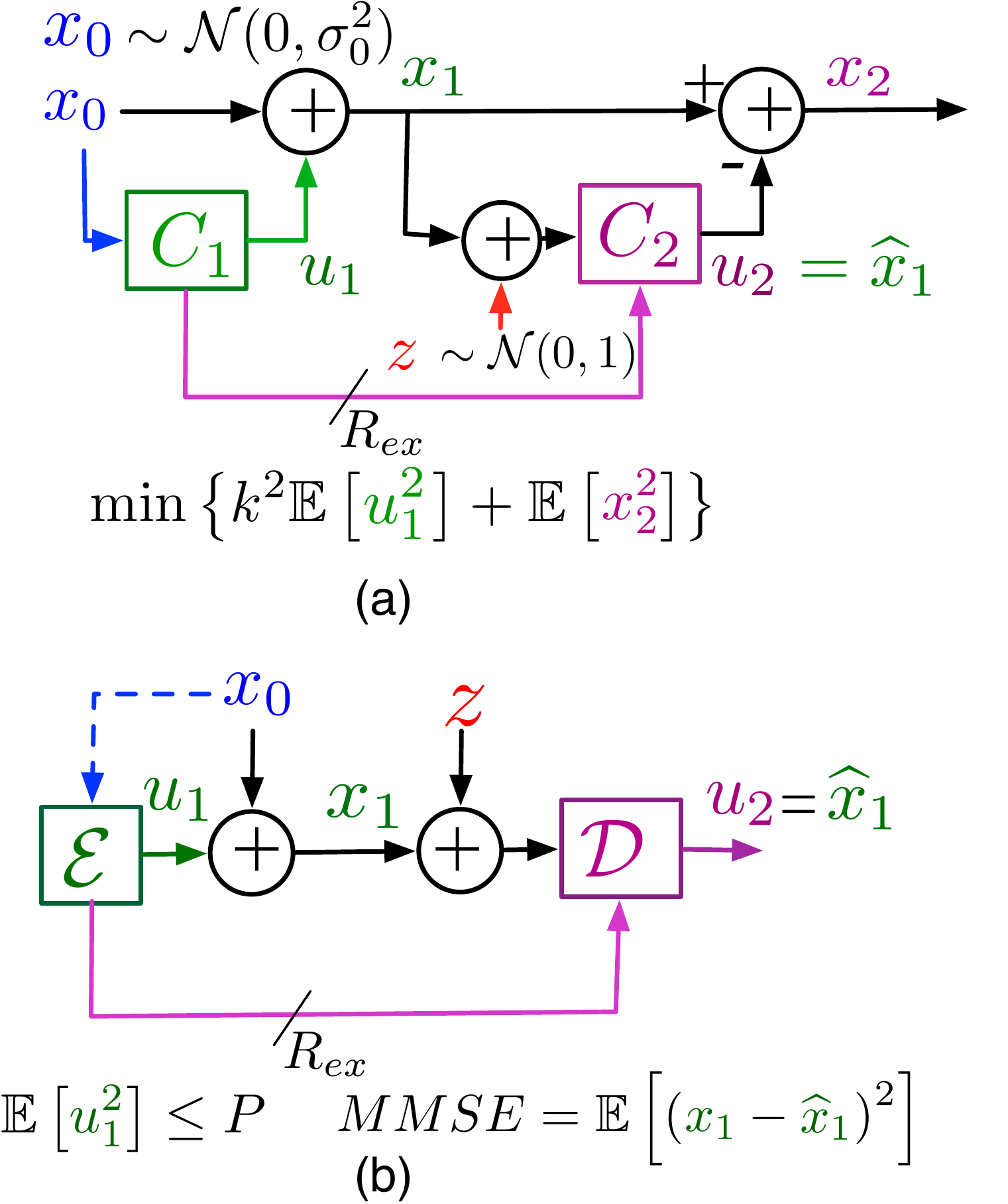} 
   \caption{The scalar version of the problem of implicit and explicit communication considered in this paper. An external channel connects the two controllers. In absence of implicit communication, the optimal strategy is linear. In absence explicit communication, an approximately-optimal strategy is quantization. Therefore, a natural strategy for this problem of implicit and explicit communication, proposed in~\cite{MartinsSideInfo}, is to communicate linearly over the external channel, and use quantization over the implicit channel. Fig.~\ref{fig:ComparisonWithNuno} shows that our binning-based synergistic strategy can outperform this natural strategy by an arbitrarily large factor.}
   \label{fig:synergy}
\end{figure}

%
%
\section{Notation and problem statement}
\label{sec:probstat}
Vectors are denoted in bold, with a superscript to denote their length (e.g. $\m{x}$ is a vector of length $m$). Upper case is used for random variables or random vectors (except when denoting power $P$), while lower case symbols represent their realizations. Hats $(\,\widehat{\cdot{}}\,)$ on the top of random variables denote the estimates of the random variables. The block-diagram for the extension of Witsenhausen's counterexample considered in this paper is shown in Fig.~\ref{fig:synergy}. $\mathcal{S}^m(r)$ denotes a sphere of radius $r$ centered at the origin in $m$-dimensional Euclidean space $\mathbb{R}^m$. $Vol(A)$ denotes volume of the set $A$ in $\mathbb{R}^m$.

A control strategy is denoted by $\gamma=(\gamma_1,\gamma_2)$, where $\gamma_i$ is the function that maps the observations at $\co{i}$ to the control inputs. The first controller observes  $\m{y}_1= \m{x}_0$ and generates a control input $\m{u}_1$ that affects the system state, and a  message $W\in\{0,1,\ldots,2^{mR}-1\}$ (that can also be viewed as a control input) for the second controller that is sent across a parallel channel.

The second controller observes $\m{y}_2=\m{x}_1 + \m{z}$, where $\m{z}$ is the disturbance, or the noise at the input of the second controller. It also observes perfectly the message $W$ sent by the first controller. The total cost is a quadratic function of the state and the input given by:
\begin{equation}
\label{eq:cost}
J^{(\gamma)}(\m{x}_0,\m{z}) = \frac{1}{m}k^2\|\m{u}_1\|^2+\frac{1}{m}\|\m{x}_2\|^2,
\end{equation}
where $\m{u}_1=\gamma_1(\m{x}_0)$, $\m{x}_2=\m{x}_0+\gamma_1(\m{x}_0)-\m{u}_2$ where $\m{u}_2 = \gamma_2(\m{x}_0+\gamma_1(\m{x}_0) + \m{z})$.  The cost expression includes a division by the vector-length $m$ to allow for natural comparisons between different vector-lengths.


Subscripts in expectation expressions denote the random variable being averaged over (e.g. $\expectp{\m{X}_0,\m{Z}_G}{\cdot{}}$ denotes averaging over the initial state $\m{X}_0$ and the test noise $\m{Z}_G$).

\section{A semi-deterministic model}
\label{sec:deterministic}
\begin{figure}[htbp] 
   \centering
   \includegraphics[width=2.5in]{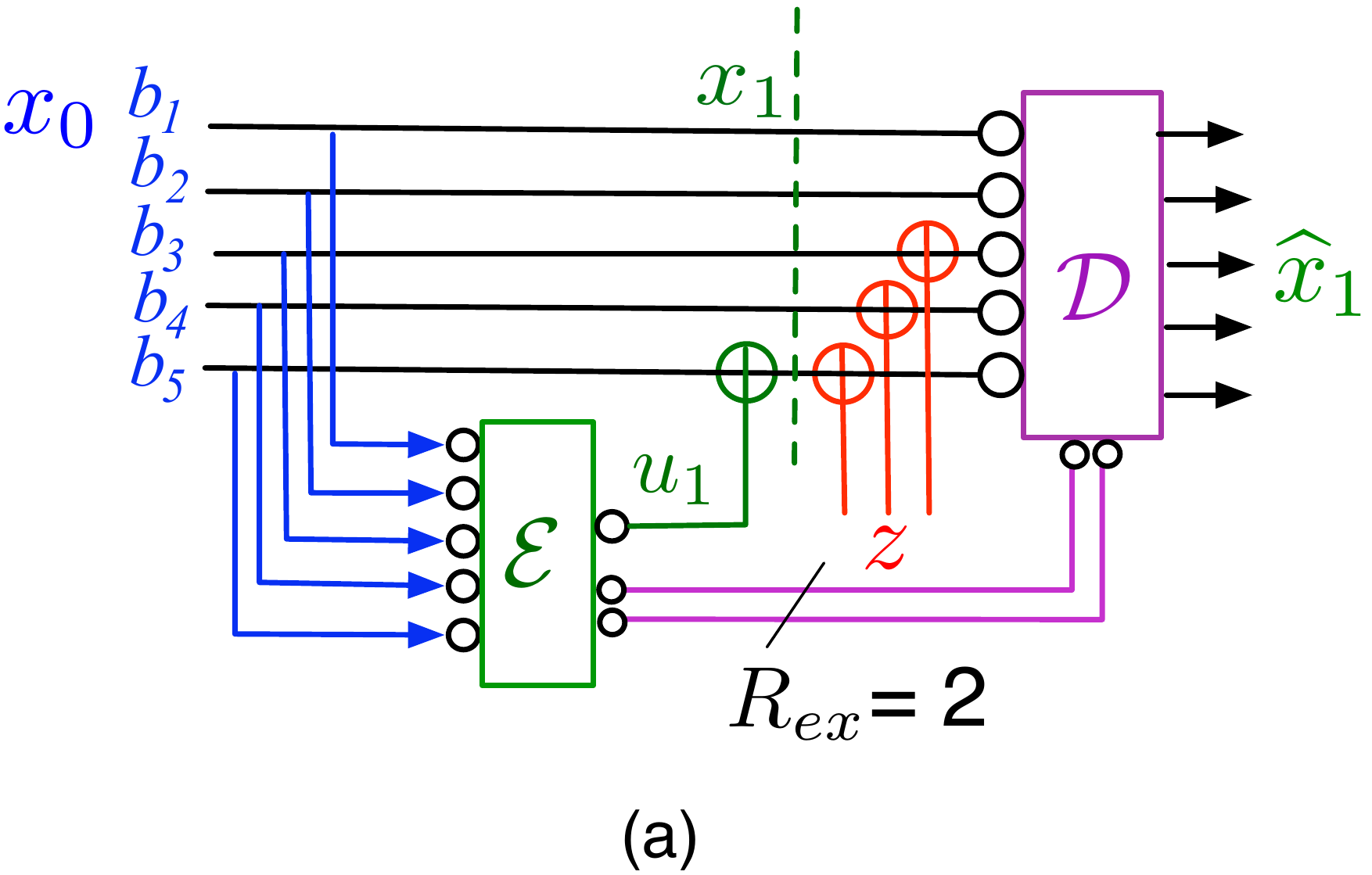}   \includegraphics[width=2.5in]{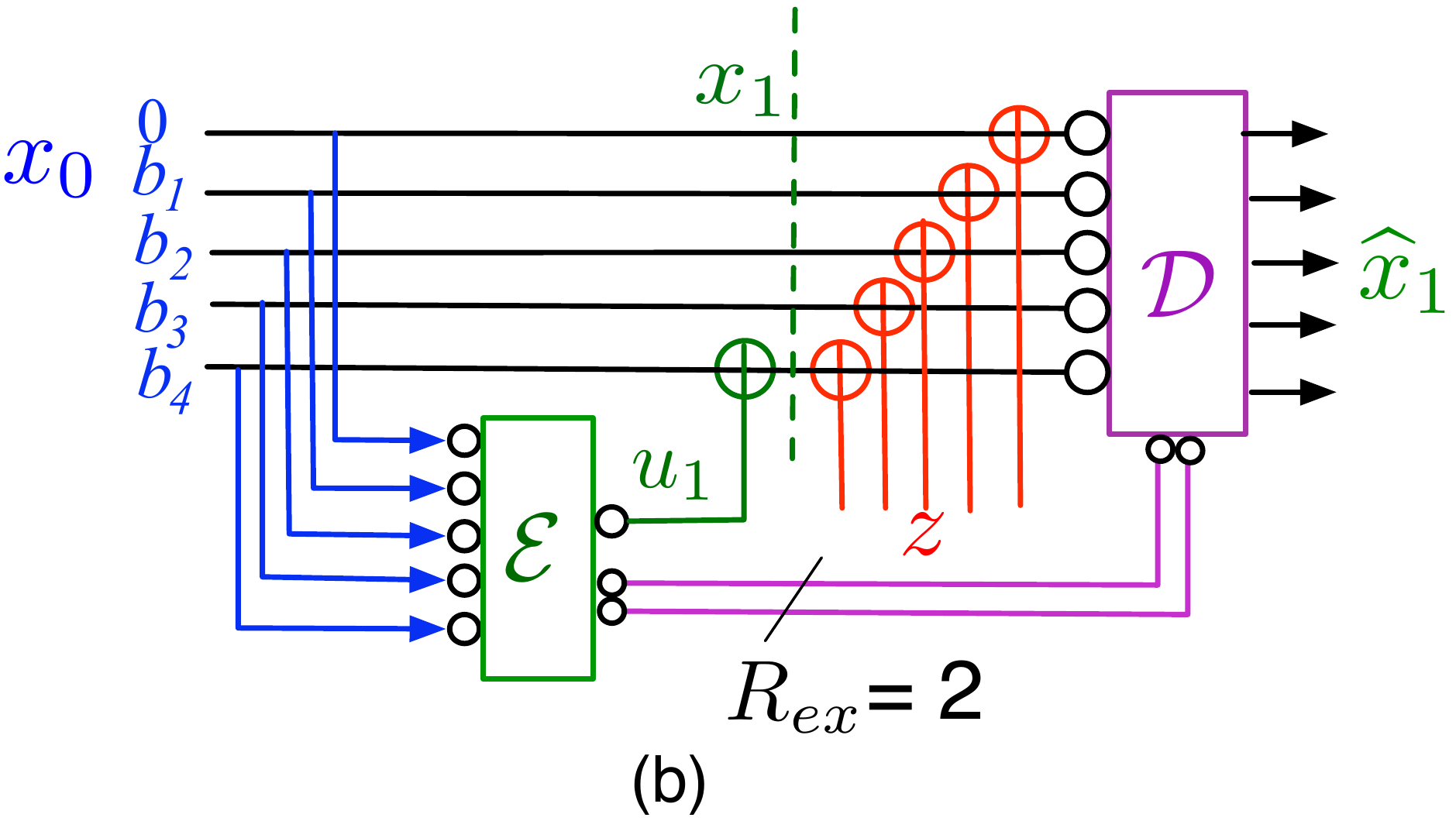}
   \caption{ A semi-deterministic model for the toy problem of implicit and explicit communication. An external channel (for this example, of capacity two bits) connects the two controllers. The case $\sigma_0^2>1$ is shown in (a), while $\sigma_0^2<1$ is shown in (b).}
   \label{fig:deterministic}
\end{figure}
We extend the deterministic abstraction of Gaussian communication networks proposed in~\cite{SalmanThesis,DeterministicModel} to a semi-deterministic model for our problem of Section~\ref{sec:probstat}. 
\begin{itemize}
\item Each system variable is represented in binary. For instance, in Fig.~\ref{fig:deterministic}, the state is represented by $b_1 b_2  b_3 . b_4b_5$, where $b_1$ is the highest order bit, and $b_5$ is the lowest.
\item The location of the decimal point is determined by the signal-to-noise ratio ($SNR$), where signal refers to the state or input to which noise is added. It is given by $\lfloor \lo{SNR} \rfloor -1$. Noise can only affect the bit before the decimal point, and the bits following it that is, $b_3$, $b_4$ and $b_5$. 
\item The power of a random variable $A$, denoted by $pow(A)$ is defined as the highest order bit that is $1$ among all the possible (binary-represented) values that $A$ can take with nonzero probability\footnote{We note that our definition of $pow(A)$ is for clarity and convenience, and is far from unique in amongst good choices.}. For instance, if $A\in\{0.01,0.11,0.1,0.001\}$, then $A$ has the power $pow(A)=0.1$.

\item Additions/subtractions in the original model are replaced by bit-wise XORs. Noise is assumed to be iid Ber(0.5).
\item The capacity of the external channel in the semi-deterministic version is the integer part (floor) of the capacity of the actual external channel.
\end{itemize}
We note here that unlike in the information-theoretic deterministic model of ~\cite{SalmanThesis}, the binary expansions in our model are valuable even after the decimal point (below noise level). Indeed, the model is not deterministic as random noise is modeled in the system\footnote{An erasure-based deterministic model for noise can instead be used. This model also has the same optimal strategies.}. This move from deterministic to semi-deterministic models is needed in decentralized control because one of the three  roles of control actions is to improve the estimability of the state when observed noisily (the other two roles being control and communication). Since smart choices of control inputs can reduce the state uncertainty in the LQG model, a simplified model should allow for this possibility as well (the matter is discussed at length in~\cite{PulkitThesis}). 

The semi-deterministic abstraction for our extension of Witsenhausen's counterexample is shown in Fig.~\ref{fig:deterministic}. 
The original cost of $k^2 u_1^2+x_2^2$ now becomes  $k^2 pow(u_1)+pow(x_2)$.  As in Fig.~\ref{fig:synergy}, the encoder for this  semi-determinisitic model observes $x_0$ noiselessly. Addition is represented by XORs, with the relative power of the terms to be added deciding which bits are affected. For instance, in Fig.~\ref{fig:deterministic}, the power of the encoder input is sufficient to only affect the last bits of the state $x_0$. The noise bits are assumed to be distributed iid Ber(0.5).

\subsection{Optimal strategies for the semi-deterministic abstraction}

We characterize the optimal tradeoff between the input power $pow(u_1)$ and the power in the MMSE error $pow(x_2)$. The minimum total cost problem is a convex dual of this problem, and can be obtained easily. Let the power of $x_0$, $pow(x_0)$ be $\sigma_0^2$. The noise power is assumed to be $1$.

\textit{Case 1}: $\sigma_0^2> 1$.\\
This case is shown in Fig.~\ref{fig:deterministic}(b). The bits $b_1,b_2$ are communicated noiselessly to the decoder, so the encoder does not need to communicate them implicitly or explicitly. The external channel has a capacity of two bits, so it can be used to communicate two of $b_3,b_4$ and $b_5$. It should be used to communicate the higher-order bits among those corrupted by noise, \textit{i.e.}, bits $b_3, b_4$. The control input $u_1$ should be used to modify the lower-order bits (bit $b_5$ in Fig.~\ref{fig:deterministic}). In the example shown, if $P<0.01$, $MMSE=0.01$, else $MMSE=0$. 


In this case (shown in Fig.~\ref{fig:deterministic}(b)), the signal power is smaller than noise power. All the bits are therefore corrupted by noise, and nothing can be communicated across the implicit channel. In order for the decoder to be able to decode any bit in the representation of $x_1$, it must either a) know the bit in advance (for instance, encoder can force the bit to $0$), or b) be communicated the bit on the external channel. Since the encoder should use minimum power, it is clear that the most significant bits of the state (bits $b_1,b_2$ in Fig.~\ref{fig:deterministic}(b)) should be communicated on the external channel. The encoder, if it has sufficient power, can then force the lower order bits ($b_3, b_4$ in Fig.~\ref{fig:deterministic}(b)) of $x_1$ to zero. In the example shown in Fig.~\ref{fig:deterministic}(b), if $P<0.001$, $MMSE=0.001$, else $MMSE=0$. 

\subsection{What scheme does the semi-deterministic model suggest over reals?}
A linear communication scheme over the external channel would correspond to communicating the highest-order bits of the state. The scheme for the semi-deterministic abstraction (Section~\ref{sec:deterministic}) communicates instead the highest order bits \textit{that are at or below the noise level}. This suggests that the external channel should not be used in a linear fashion --- the higher order bits are already known at the decoder. Instead, the external channel should be used to  communicate bits that are corrupted by noise  --- more refined information about the state that is not already implicitly communicated by the noisy state itself.

The resulting scheme for the problem over reals is illustrated in Fig.~\ref{fig:binning}. The encoder forces lower order bits of the state to zero, thereby truncating the binary expansion, or effectively quantizing the state into bins. The higher order bits that are corrupted by noise ($b_3,b_4$ in Fig.~\ref{fig:deterministic}(a)) are communicated via the external channel. These bits can be thought of as representing the color, \textit{i.e.} the bin index, of quantization bins, where set of $2^{R_{ex}}$ consecutive quantization-bins are labelled with $2^{R_{ex}}$ colors with a fixed order (with zero, for instance, colored blue). The bin-index associated with the color of the bin is sent across the external channel. The decoder finds the quantization point nearest to $y_2$ that has the same bin-index as that received across the external channel.

The scheme is very similar to the binning scheme used for Wyner-Ziv coding of a Gaussian source with side information~\cite{WynerZiv}, which is not surprising because of similarity of our problem with the Wyner-Ziv formulation.

\begin{figure}[htbp] 
   \centering
\includegraphics[width=3.5in]{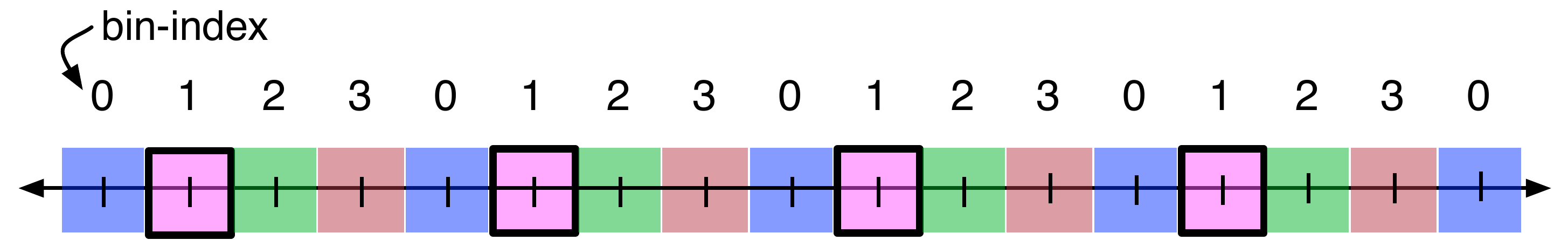} 
   \caption{ The strategy intuited from the semi-deterministic model naturally yields a binning-based strategy for reals that leads to a synergistic use of implicit and explicit communication. The external channel get the decoder the bin-index (in this example, the index is $1$). The more significant bits (coarse bin) is received from the implicit channel. Effectively,  use of the external channel increases the distance between the `valid' codewords by a factor of $2^{R_{ex}}$. }
   \label{fig:binning}
\end{figure}

\section{Gaussian external channel}
\label{sec:gaussian}
A more realistic model of the external channel is a power constrained additive Gaussian noise channel, which was considered in~\cite{Shoarinejad,MartinsSideInfo}. Without loss of generality, we assume that the noise in the external channel is also of variance $1$. 

At finite-lengths, an upper bound can be calculated using binning-based strategies. This binning-strategy turns out to outperform Martins's strategy by a factor that diverges to infinity. The key is to choose the set of problems where the initial state variance and the power on the external channel, denoted by $P_{ex}$, are almost equal. In this case, a strategy that communicates the state on the external channel is not helpful --- implicit channel can communicate the state at almost the same fidelity. Fig.~\ref{fig:ComparisonWithNuno} shows that fixing the relation $P_{ex}=\sigma_0^2$, as $\sigma_0^2\rightarrow\infty$, the ratio of costs attained by the binning strategy to that attained by Martins's strategy diverges to infinity.

\begin{figure}[htbp] 
   \centering
   \includegraphics[width=3in]{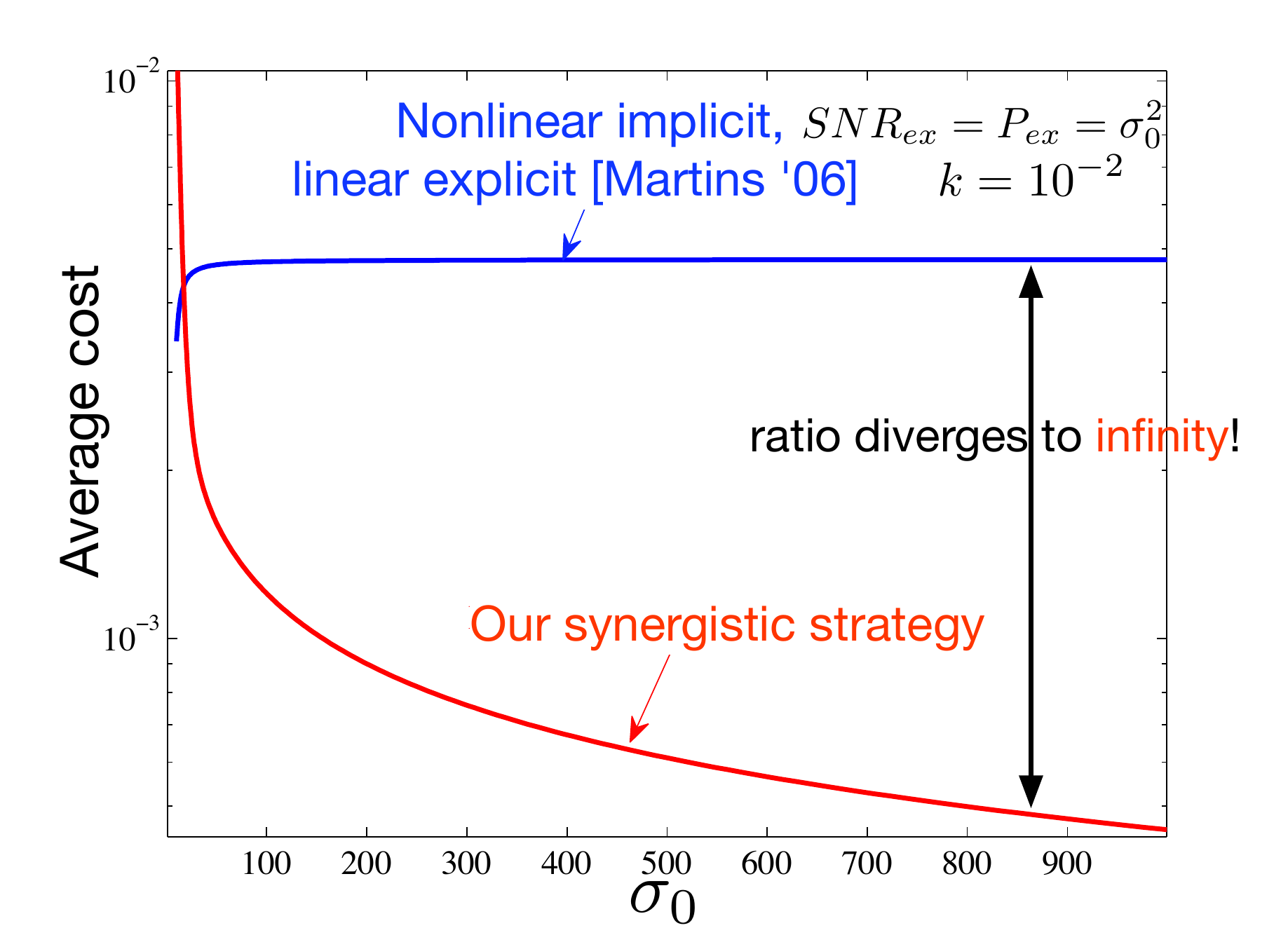} 
   \caption{If the SNR on the external channel is made to scale with SNR of the initial state, then our binning-based strategy outperforms strategy in~\cite{MartinsSideInfo} by a factor that diverges to infinity.}
   \label{fig:ComparisonWithNuno}
\end{figure}

\section{Asymptotic and scalar versions of the problem}
\subsection{Asymptotic version}
\label{sec:asymptotic}
We now show that the binning strategy of Section~\ref{sec:deterministic} is approximately-optimal in the limit of infinite-lengths.


%
\begin{theorem}
\label{thm:finiteratiosynergy}
For the extension of Witsenhausen's counterexample with an external channel connecting the two controllers,
\begin{eqnarray*}
&& \inf_{P\geq 0} k^2P + \left(\left(   \sqrt{\kappa_{new} } - \sqrt{P}  \right)^+ \right)^2 \\
& \leq & \cost{}_{opt} \leq \mu  \left(\inf_{P\geq 0}k^2P + \left(\left(   \sqrt{\kappa_{new} } - \sqrt{P}  \right)^+ \right)^2\right),
\end{eqnarray*}
where $\mu \leq 64$, $\kappa_{new}= \frac{\sigma_0^2 2^{-2R_{ex}}}{   \overline{P}  + 1   }$, where $\overline{P}=\left(\sigma_0+\sqrt{P}\right)^2$and the upper bound is achieved by binning-based quantization strategies. Numerical evaluation shows that $\mu <8$.

\end{theorem}

\begin{proof}
\textit{Lower bound} \\
We need the following lemma from~\cite[Lemma 3]{WitsenhausenJournal}.
\begin{lemma}
\label{lem:triangle}
For any three random vectors $A$, $B$ and $C$,
\begin{equation*}
\sqrt{\expect{\|B-C\|^2}}\geq \sqrt{\expect{\|A-C\|^2}} - \sqrt{\expect{\|A-B\|^2}}.
\end{equation*}
\end{lemma}
\begin{proof}
See~\cite{WitsenhausenJournal}.
\end{proof}
Substituting $\m{X}_0$ for $A$, $\m{X}_1$ for $B$, and $\m{U}_2$ for $C$ in Lemma~\ref{lem:triangle},
\begin{eqnarray}
\label{eq:sqrteqn}
\nonumber  && \sqrt{\expect{\|\m{X}_1-\m{U}_2\|^2}}\\
&& \geq  \sqrt{\expect{\|\m{X}_0-\m{U}_2\|^2}}-\sqrt{\expect{\|\m{X}_0-\m{X}_1\|^2}}.
\end{eqnarray}
We wish to lower bound  $\expect{\|\m{X}_1-\m{U}_2\|}$. The second term on the RHS is smaller than $\sqrt{mP}$. Therefore, it suffices to lower bound the first term on the RHS of~\eqref{eq:sqrteqn}. 

With what distortion can $\m{x}_0$ be communicated to the decoder? The capacity of the parallel channel is the sum of the two capacities $C_{sum}=R_{ex}+C_{implicit}$. The capacity $C_{implicit}$ is upper bounded by $\frac{1}{2}\lo{1+ \overline{P}}$ where $\overline{P}: = \left(\sigma_0 + \sqrt{P}\right)^2$. Using Lemma~\ref{lem:triangle}, the distortion in reconstructing $\m{x}_0$ is lower bounded by
\begin{eqnarray*}
D(C_{sum}) & = & \sigma_0^2 2^{-2C_{sum}}= \sigma_0^2 2^{-2 R_{ex} - 2C_{implicit}}\\
& \geq & \frac{\sigma_0^2 2^{-2R_{ex}}}{   \overline{P}  + 1   }=\kappa_{new}.
\end{eqnarray*}
Thus the distortion in reconstructing $\m{x}_1$ is lower bounded by
\begin{eqnarray*}
\left(\left(   \sqrt{\kappa_{new} } - \sqrt{P}  \right)^+ \right)^2.
\end{eqnarray*}
This proves the lower bound in Theorem~\ref{thm:finiteratiosynergy}.\\
\textit{Upper bound}\\
\textbf{Quantization}: 
This strategy is used for $\sigma_0^2>1$. Quantize $\m{x}_0$ at rate $C_{sum}=R_{ex}+C_{implicit}$. Bin the codewords randomly into $2^{nR_{ex}}$ bins, and send the bin index on the external channel. On the implicit channel, send the codeword closest to the vector $\m{x}_0$.

The decoder looks at the bin-index on the external channel, and keeps only the codewords that correspond to the bin index. This subset of the codebook, which now corresponds to the set of valid codewords, has rate $C_{implicit}$. The required power $P$ (which is the same as the distortion introduced in the source $\m{x}_0$) is thus given by
\begin{eqnarray*}
\frac{1}{2}\lo{\frac{\sigma_0^2}{P}} \leq R_{ex} + \frac{1}{2}\lo{1+\sigma_0^2 - P},
\end{eqnarray*}
which yields the solution $P = \frac{(1+\sigma_0^2) - \sqrt{  (1+\sigma_0^2)^2 - 4\sigma_0^2 2^{-2R_{ex}}  }}{2}$ which is smaller than $1$.  Thus,
\begin{eqnarray*}
P & = &\frac{(1+\sigma_0^2) - \sqrt{  (1+\sigma_0^2)^2 - 4\sigma_0^2 2^{-2R_{ex}}  }}{2}\\
& = & \frac{1}{2}(1+\sigma_0^2)\left(1 - \sqrt{1 - 4\frac{\sigma_0^2}{ (1+\sigma_0^2)^2} 2^{-2R_{ex}}  }\right).
\end{eqnarray*}
Now note that $\frac{\sigma_0^2}{(1+\sigma_0^2)^2}$ is a decreasing function of $\sigma_0^2$ for $\sigma_0^2>1$. Thus, $\frac{\sigma_0^2}{(1+\sigma_0^2)^2}<\frac{1}{4}$ for $\sigma_0^2>1$, and $1 - 4\frac{\sigma_0^2}{ (1+\sigma_0^2)^2} 2^{-2R_{ex}}>0$. Because $0<1 - 4\frac{\sigma_0^2}{ (1+\sigma_0^2)^2} 2^{-2R_{ex}}  <1$, 
\begin{eqnarray*}
\sqrt{1 - 4\frac{\sigma_0^2}{ (1+\sigma_0^2)^2} 2^{-2R_{ex}}  } \geq 1 - 4\frac{\sigma_0^2}{ (1+\sigma_0^2)^2} 2^{-2R_{ex}},
\end{eqnarray*}
and therefore 
\begin{eqnarray*}
P & \leq  & \frac{1}{2}(1+\sigma_0^2)\left(1 - \left(1 - 4\frac{\sigma_0^2}{ (1+\sigma_0^2)^2} 2^{-2R_{ex}}  \right)\right)\\
& = & \frac{1}{2}(1+\sigma_0^2)\left( 4\frac{\sigma_0^2}{ (1+\sigma_0^2)^2} 2^{-2R_{ex}} \right)\\
& = & \frac{2\sigma_0^2}{1+\sigma_0^2}2^{-2R_{ex}} \leq 2\times 2^{-2R_{ex}}.
\end{eqnarray*}

The other strategies that complement this binning strategy are the analogs of zero-forcing and zero-input. 

\textbf{Analog of the zero-forcing strategy}
The state $\m{x}_0$ is quantized using a rate-distortion codebook of $2^{mR_{ex}}$ points. The encoder sends the bin-index of the nearest quantization-point on the external channel. Instead of forcing the state all the way to zero, the input is used to force the state to the nearest quantization point. The required power is given by the distortion $\sigma_0^2 2^{-2R_{ex}}$. The decoder knows exactly which quantization point was used, so the second stage cost is zero. The total cost is therefore $k^2\sigma_0^2 2^{-2R_{ex}}$.


\textbf{Analog of Zero-input strategy}

\textit{Case 1}: $\sigma_0^2\leq 4$.

Quantize the space of initial state realizations using a random codebook of rate $R_{ex}$, with the codeword elements chosen i.i.d $\mathcal{N}(0,\sigma_0^2(1-2^{-2R_{ex}}))$. Send the index of the nearest codeword on the external channel, and ignore the implicit channel. The asymptotic achieved distortion is given by the distortion-rate function of the Gaussian source $\sigma_0^22^{-2R_{ex}}$.

\textit{Case 2}: $R_{ex}\leq 2$. Do not use the external channel. Perform an MMSE operation at the decoder on the state $\m{x}_0$. The resulting error is $\frac{\sigma_0^2}{\sigma_0^2+1}$. 

\textit{Case 3}: $\sigma_0^2>4, R_{ex}>2$.

Our proofs in this part follow those in~\cite{Tse2005}. Let $R_{code}= R_{ex}+\frac{1}{2}\lo{\frac{\sigma_0^2}{3}}-\epsilon$. A codebook of rate $R_{code}$ is designed as follows. Each codeword is chosen randomly and uniformly inside a sphere centered at the origin and of radius $m\sqrt{\sigma_0^2-D}$, where $D=\sigma_0^2 2^{-2R_{code}}=3\times 2^{-2(R_{ex}-\epsilon)}$. This is the attained asymptotic distortion when the codebook is used to represent\footnote{In the limit of infinite block-lengths, average distortion attained by a uniform-distributed random-codebook and a Gaussian random-codebook of the same variance is the same~\cite{Tse2005}.} $\m{x}_0$.

Distribute the $2^{mR_{code}}$ points randomly into $2^{mR_{ex}}$ bins that are indexed $\{1,2,\ldots{},2^{mR_{ex}}\}$. The encoder  chooses the codeword $\m{x}_{code}$ that is closest to the initial state. It sends the bin-index (say $i$) of the codeword across the external channel. 

Let $\m{z}_{code}=\m{x}_0-\m{x}_{code}$. The received signal $\m{y}_2=\m{x}_0+\m{z}=\m{x}_{code}+\m{z}_{code}+\m{z}$, which can be thought of as receiving a noisy version of codeword $\m{x}_{code}$ with a total noise of variance $D+1$, since $\m{z}_{code}\independent \m{z}$. 

The decoder receives the bin-index $i$ on the external channel. Its goal is to find $\m{x}_{code}$. It looks for a codeword from bin-index $i$ in a sphere of radius $D+1+\epsilon$ around $\m{y}_2$. We now show that it can find $\m{x}_{code}$ with probability converging to $1$ as $m\rightarrow\infty$. A rigorous proof that MMSE also converges to zero can be obtained along the lines of proof in~\cite{WitsenhausenJournal}. 

To prove that the error probability converges to zero, consider the total number of codewords that lie in the decoding sphere. This, on average, is bounded by 
\begin{eqnarray*}
&\frac{2^{mR_{code}}} {Vol\left(\mathcal{S}^m\left(m\sqrt{(\sigma_0^2-D+\epsilon)}\right)\right)}   Vol\left(\mathcal{S}^m\left(m\sqrt{D+1 +\epsilon}\right)\right)\\
& =  \frac{2^{m\left(R_{ex}-\epsilon+\frac{1}{2}\lo{\frac{\sigma_0^2}{3}}\right)}} {Vol\left(\mathcal{S}^m\left(m\sqrt{(\sigma_0^2-D+\epsilon)}\right)\right)}    Vol\left(\mathcal{S}^m\left(m\sqrt{D+1 +\epsilon}\right)\right)\\
& =  \frac{2^{m\left(R_{ex}-\epsilon+\frac{1}{2}\lo{\frac{\sigma_0^2}{3}} \right)}} {\left(m\sqrt{\sigma_0^2-D+\epsilon}\right)^m} \left(m\sqrt{D+1 +\epsilon}\right)^m \\
& =  2^{m(R_{ex}-\epsilon)} 2^{\left(\frac{m}{2}\lo{\frac{\sigma_0^2(D+1+\epsilon)}{3(\sigma_0^2-D+\epsilon)}} \right)}.
\end{eqnarray*}
Let us pick another codeword in the decoding sphere. Probability that this codeword has index $i$ is $2^{-mR_{ex}}$. Using union bound, the probability that there exists another codeword in the decoding sphere of index $i$ is bounded by  
\begin{eqnarray*}
&&2^{-mR_{ex}}2^{m(R_{ex}-\epsilon)} 2^{\left(\frac{m}{2}\lo{\frac{\sigma_0^2(D+1+\epsilon)}{3(\sigma_0^2-D+\epsilon)}} \right)}\\
&=& 2^{-m\epsilon}2^{\left(\frac{m}{2}\lo{\frac{\sigma_0^2(D+1+\epsilon)}{3(\sigma_0^2-D+\epsilon)}} \right)}.
\end{eqnarray*}
It now suffices to show that the second term converges to zero as $m\rightarrow\infty$. Since $D=3\times 2^{-2(R_{ex}-\epsilon )}$. Since $R_{ex}>2$, $D<\frac{3}{4}\times 2^{\epsilon}<\frac{5}{6}-\epsilon $ for small enough $\epsilon$. Since $\sigma_0^2>4$, $D<\frac{5}{6}\frac{\sigma_0^2}{4}<\frac{\sigma_0^2}{4} + \epsilon$,
\begin{eqnarray*}
\frac{\sigma_0^2(D+1+\epsilon)}{3(\sigma_0^2-D+\epsilon)} < \frac{\sigma_0^2\times\left(\frac{5}{6}+1\right)}{ 3\frac{3\sigma_0^2}{4}  }=\frac{\frac{11}{6}}{\frac{9}{4}}=\frac{22}{27}<1.
\end{eqnarray*}
Thus the cost here is bounded by $3\times 2^{-2(R_{ex}-\epsilon)}$ which is bounded by $4\times 2^{-2R_{ex}}$ for small enough $\epsilon$. 

\subsubsection{Bounded ratios for the asymptotic problem}
The upper bound is the best of the vector-quantization bound, $2k^22^{-2R_{ex}}$, zero-forcing $k^2\sigma_0^22^{-2R_{ex}}$, and zero-input bounds of $\sigma_0^2 2^{-2R_{ex}}$ and $4\times 2^{-2R_{ex}}$. 

\textit{Case 1}: $P^*>\frac{2^{-2R_{ex}}}{16}$.\\
In this case, the lower bound is larger than $k^2 \frac{2^{-2R_{ex}}}{16}$. Using the upper bound of $4\times 2^{-2R_{ex}}$, the ratio is smaller than $64$. 

\textit{Case 2}: $P^*\leq\frac{2^{-2R_{ex}}}{16},\sigma_0^2\geq 1$.\\
Since $R_{ex}\geq 0$, $P^*\leq\frac{1}{16}$. Thus,
\begin{eqnarray*}
\kappa_{new}=\frac{\sigma_0^2 2^{-2R_{ex}}}{(\sigma_0+\sqrt{P^*})^2+1}>\frac{1}{\left(1+\frac{1}{4}\right)^2+1}= \frac{16}{41}2^{-2R_{ex}}.
\end{eqnarray*}
Thus, the lower bound is greater than the $MMSE$ which is larger than
\begin{equation}
\left(  \sqrt{\frac{16}{41}}  -\sqrt{\frac{1}{16}}  \right)^22^{-2R_{ex}}\approx 0.14\times 2^{-2R_{ex}}.
\end{equation}
Using the upper bound of $4\times 2^{-2R_{ex}}$, the ratio is smaller than $29$.

\textit{Case 3}: $P^*\leq\frac{2^{-2R_{ex}}}{16},\sigma_0^2< 1$.\\
If $P^*>\frac{\sigma_0^22^{-2R_{ex}}}{25}$, using the upper bound of $\sigma_0^22^{-2R_{ex}}$, the ratio is smaller than $25$.

If $P^*\leq \frac{\sigma_0^22^{-2R_{ex}}}{25}<\frac{1}{25}$, 
\begin{eqnarray*}
\kappa_{new} &=& \frac{\sigma_0^22^{-2R_{ex}}}{(\sigma_0+\sqrt{P^*})^2+1}\\
& \geq & \frac{\sigma_0^22^{-2R_{ex}}}{\left(1 + \frac{1}{5}\right)^2+1}\sigma_0^22^{-2R_{ex}} = \frac{25}{61}\sigma_0^22^{-2R_{ex}}.
\end{eqnarray*} 
 Thus, a lower bound on $MMSE$, and hence also on the total costs, is
 \begin{eqnarray*}
 \left(\sqrt{\frac{25}{61}}-\sqrt{\frac{1}{25}}\right)^2\sigma_0^22^{-2R_{ex}}\approx 0.19 \sigma_0^22^{-2R_{ex}}.
 \end{eqnarray*}
Using the upper bound of $\sigma_0^22^{-2R_{ex}}$, the ratio is smaller than $\frac{1}{0.19}<6$. 

\begin{figure}[htbp] 
   \centering
   \includegraphics[width=3.4in]{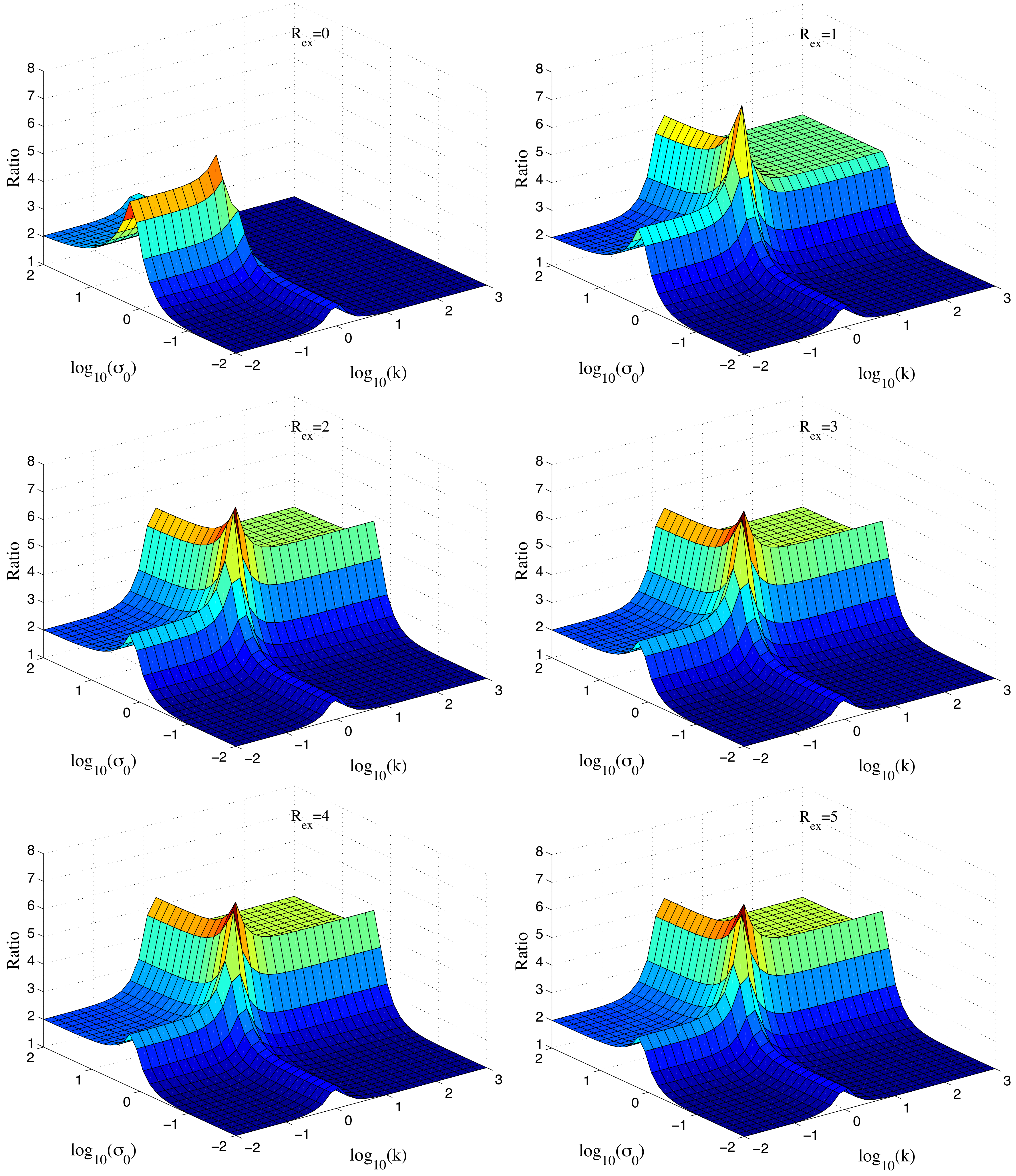} 
   \caption{The ratio of upper and lower bounds for the asymptotic problem are bounded by a factor of $8$ for all $k,\sigma_0$ and $R_{ex}$.}
   \label{fig:vector}
\end{figure}

Numerical evaluations, shown in Fig.~\ref{fig:vector}, show that the ratio is smaller than $8$.
\end{proof}
\subsection{Scalar case}
\label{sec:scalar}
We first derive a lower bound for finite-vector lengths.
The obtained bounds are tighter than those in Theorem~\ref{thm:finiteratiosynergy} and depend explicitly on the  vector length $m$.
\begin{theorem}[Refined lower bound for finite-lengths]
\label{thm:newbound}
For a finite-dimensional vector version of the problem, if for a strategy $\gamma(\cdot{})$ the average power $\frac{1}{m}\expectp{\m{X}_0}{\|\m{U}_1\|^2}=P$, the following lower bound holds on the second stage cost for any choice of $\sigma_G^2\geq 1$ and $L>0$
\begin{equation*}
\overline{J}_2^{(\gamma)}(m,k^2,\sigma_0^2) \geq \eta(P,\sigma_0^2,\sigma_G^2,L).
\end{equation*}
where $\eta(P,\sigma_0^2,\sigma_G^2,L)=$
\begin{eqnarray*}
&&\frac{\sigma_G^m}{c_m(L)}\exp\left(-\frac{mL^2(\sigma_G^2-1)}{2}\right)\\
&&\times\left(  \left(\sqrt{\kappa_2(P,\sigma_0^2,\sigma_G^2,L)} - \sqrt{P}\right)^+   \right)^2,
\end{eqnarray*}
where $\kappa_2(P,\sigma_0^2,\sigma_G^2,L):=$
\begin{eqnarray*}
\frac{\sigma_0^2\sigma_G^22^{-2R_{ex}}}{c_m^{\frac{2}{m}}(L)e^{1-d_m(L)}\left((\sigma_0+\sqrt{P})^2+d_m(L)\sigma_G^2\right)},
\end{eqnarray*}
 $c_m(L):=\frac{1}{\Pr(\|\m{Z}\|^2\leq mL^2)}= \left(1-\psi(m,L\sqrt{m})\right)^{-1}$, 
$d_m(L):=\frac{\Pr(\|\mk{Z}{m+2}\|^2\leq mL^2)}{\Pr(\|\m{Z}\|^2\leq mL^2)} =
\frac{1-\psi(m+2,L\sqrt{m})}{1-\psi(m,L\sqrt{m})}$, \\$0< d_m(L)<1$, and 
$\psi(m,r)=\Pr(\|\m{Z}\|\geq r)$.
Thus the following lower bound holds on the total cost
\begin{equation}
\overline{J}_{\min}(m,k^2,\sigma_0^2) \geq \inf_{P\geq 0} k^2P + 
\eta(P,\sigma_0^2,\sigma_G^2,L),
\end{equation}
for any choice of $\sigma_G^2\geq 1$ and $L>0$ (the choice can depend on $P$). Further, these bounds are at least as tight as those of Theorem~\ref{thm:finiteratiosynergy} for all values of $k$ and $\sigma_0^2$.
\end{theorem}
\begin{proof}
We remark that the only difference in this lower bound as compared to that in~\cite{FiniteLengthsWitsenhausen} is the term for $R_{ex}$ in the expression for $\kappa_2$. The proof follows along the lines of that of~\cite[Theorem 3]{FiniteLengthsWitsenhausen}. See Appendix~\ref{app:finitelower} for the proof.
\end{proof}

\begin{theorem}[Upper bound for scalar case]
\label{thm:scalarupper}
An upper bound on costs for the scalar case is given by
\begin{eqnarray*}
&\cost{}_{opt}\leq \min\{ \inf_{P\geq 2^{-2R_{ex}}}k^2 P + \psi(3,2^{2R_{ex}}P),\\
&  ck^2\sigma_0^2, c \sigma_0^22{-2R_{ex}},\\
& a^22^{-2R_{ex}}+(1+a)^2 e^{-\frac{a^2}{2}+\frac{3}{2}(1+\lon{a^2})} \},
\end{eqnarray*}
where\footnote{This upper bound on $c$ is the believed upper bound on the distortion-rate function $D_s(R)=c\sigma_0^22^{-2R}$ of a scalar Gaussian source. We have been unable to find a rigorous proof of this result, although the result is known to holds at high rates~\cite{PanterDite}, and Lloyd's empirical results~\cite[Table VIII]{lloyd} suggest that the bound holds for all rates.} $c\leq 2.72$ and $\psi(m,r)$ is defined in Theorem~\ref{thm:newbound}. 
\end{theorem}
\begin{proof}
Just as for the asymptotic case, each term in the upper bound corresponds to a certain strategy.\\
\textbf{Quantization}\\
Divide the real line into uniform quantization bins of size $\sqrt{P}$. The quantization points are located at the center of these bins. Number consecutive bins $i (\text{mod}\; 2^{R_{ex}})$ starting with bin $0$ which contains the origin. The encoder forces the initial state to the quantization point closest to the initial state, requiring a power of at most $P$. It also sends the index of the quantization bin on the external channel.

The decoder looks at the bin-index, and finds the nearest quantization point corresponding to the particular bin-index. The resulting MMSE error is given by $\expect{z^2\indi{|z|>2^{R_{ex}} \sqrt{P}}}$. This is shown to equal $\psi(3,2^{R_{ex}} \sqrt{P})$ in~\cite{FiniteLengthsWitsenhausen}. This yields the first term.

\textbf{Analog of zero-forcing}\\
Quantize the real-line using a quantization codebook of rate $R_{ex}$. The encoder forces $x_0$ to the nearest quantization point, and sends the index of the point to the decoder. The distortion is  bounded by $2.72\sigma_0^22^{-2R_{ex}}$~\cite{GrayNeuhoff}. The decoder has a perfect estimate of $x_1$, thus the total cost is given by $k^2c\sigma_0^22^{-2R_{ex}}$.

\textbf{Analog of zero-input}\\
As for the asymptotic case, we break this case into two strategies. For $\sigma_0^2\leq 4$, we again use a quantization codebook of rate $R_{ex}$, but instead of zero-forcing the state, we take the distortion hit at the decoder. The resulting cost is $c\sigma_0^22^{-2R_{ex}}$. 

For $\sigma_0^2>4$, we use a construct based on the idea of sending coarse information across the implicit channel, and fine information across the explicit channel. Divide the entire line into coarse quantization-bins of size $2a$. Divide each bin into $2^{R_{ex}}$ sub-bins, each of size $2a2^{-R_{ex}}$. Number each of the sub-bins in any sub-bin from $0,1,\ldots,2^{R_{ex}}$. 

The encoder send the index of the sub-bin in which $x_0$ lies across the external channel. The decoder decodes this sub-bin by finding the nearest sub-bin to the received output that has the same index as that received across the external channel. 

If the decoder decodes the correct sub-bin, the error is bounded by $a^22^{-2R_{ex}}$. In the event when there is an error in decoding of the sub-bin, the error is bounded by $(|z|+a)^2$, which averaged under the error event $|z|>a$ takes exactly the form of~\cite[Lemma 1]{FiniteLengthsWitsenhausen}. Using that lemma, the MMSE in the error-event is bounded by
\begin{eqnarray*}
\expect{(|z|+a)^2\indi{|z|>a}} &\leq& (\sqrt{\psi(3,a)}+a\sqrt{\psi(1,a)})^2\\
& \overset{a>1}{\leq} &(1+a)^2 e^{-\frac{a^2}{2}+\frac{3}{2}(1+\lon{a^2})}.
\end{eqnarray*}
Thus the total $MMSE$ is bounded by
\begin{eqnarray*}
MMSE\leq a^22^{-2R_{ex}}+(1+a)^2 e^{-\frac{a^2}{2}+\frac{3}{2}(1+\lon{a^2})}.
\end{eqnarray*}
\end{proof}

\begin{figure}[htbp] 
   \centering
   \includegraphics[width=3.55in]{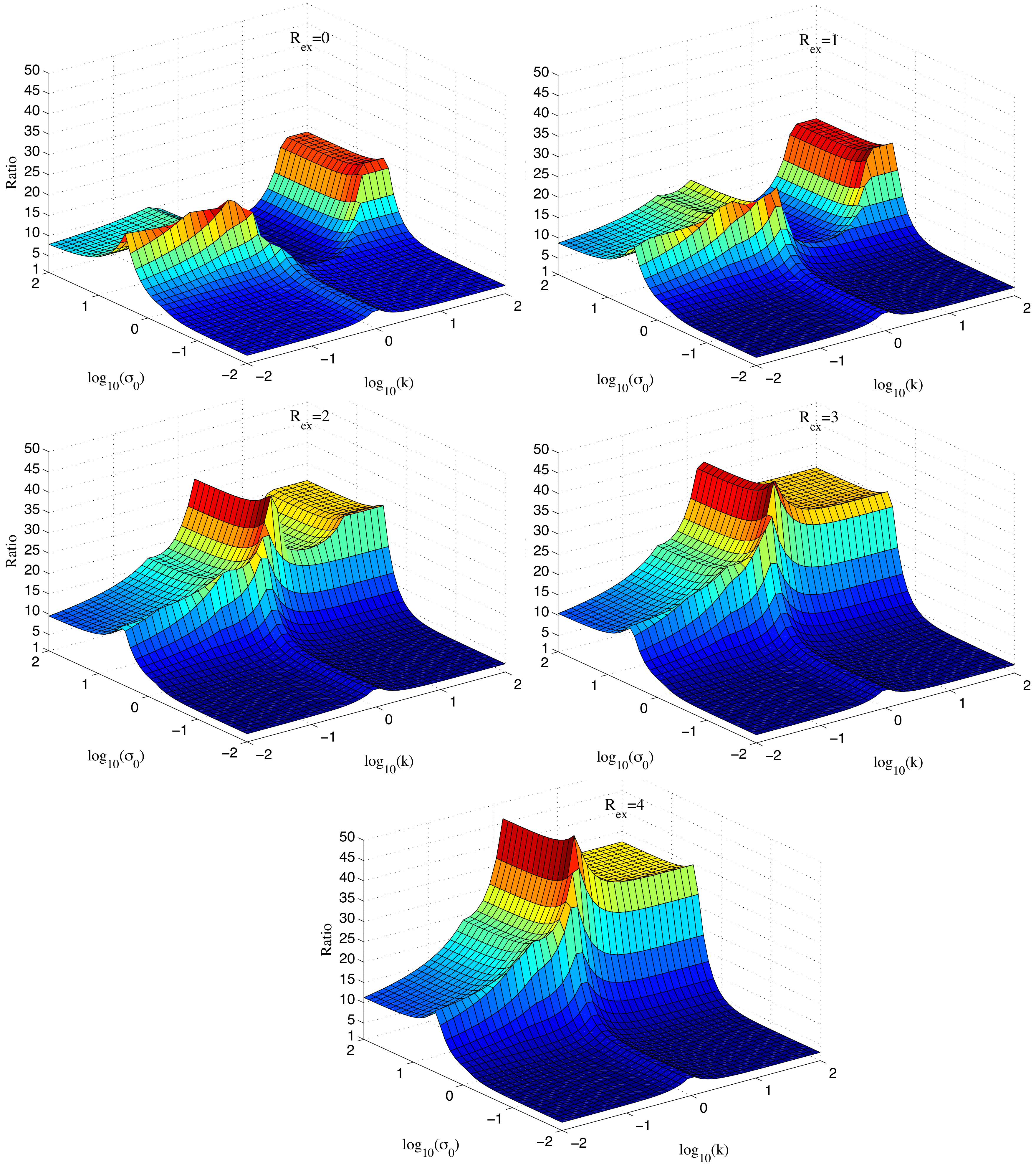} 
   \caption{Ratio of upper and lower bounds on the scalar problem for various values of $R_{ex}$. The ratio diverges to infinity as $R_{ex}\rightarrow\infty$.}
   \label{fig:scalar}
\end{figure}

Fig.~\ref{fig:scalar} shows ratio of upper bound of Theorem~\ref{thm:scalarupper} and lower bound of Theorem~\ref{thm:newbound} in the $(k,\sigma_0)$-parameter space. Even though the ratio is bounded for each $R_{ex}$, it blows up as $R_{ex}\rightarrow\infty$. 

\vspace{-0.1in}

\section{Discussions and conclusions}
\label{sec:conclusions}
The asymptotic result in Section~\ref{sec:asymptotic} extends easily to an asymptotic version of problem with a Gaussian external channel (Section~\ref{sec:gaussian}). This is because the error probability on the external channel converges to zero as the vector length $m\rightarrow\infty$ for any $R_{ex}<C_{ex}$, the capacity of the external channel, making it behave like a fixed rate external channel. Using large-deviation techniques, there is hope that the scalar problem with Gaussian external channel may also be solved approximately. 


A rate-limited noiseless channel can be thought of as a model for limited-memory controllers. The problem of Fig.~\ref{fig:synergy} can then be interpreted as a single controller system with finite memory. The problem problem considered here is also a toy problem that can design strategies for finite-memory controller problems.


\vspace{-0.1in}

\section*{Acknowledgments}

\vspace{-0.05in}

We would like to acknowledge most stimulating discussions with Tamer Ba\c{s}ar while writing this paper. We also thank Aditya Mahajan for references, and Gireeja Ranade, Pravin Varaiya and Jiening Zhan for helpful discussions. Support of grant NSF CNS-0932410 is gratefully acknowledged. 

\vspace{-0.05in}

\appendices
\section{Proof of lower bound for finite-length problem}
\label{app:finitelower}
\begin{proof}
From Theorem~\ref{thm:finiteratiosynergy}, for a given $P$, a lower bound on the average second stage cost is $\left(\left( \sqrt{\kappa_{new}}-\sqrt{P}   \right)^+\right)^2$. We derive
 another lower bound that is equal to the 
 expression for $\eta(P,\sigma_0^2,\sigma_G^2,L)$. 

 Define $\mathcal{S}_L^G:=\{\m{z}:\|\m{z}\|^2\leq mL^2\sigma_G^2\}$
 and use subscripts to denote which probability model is being used for the second stage observation noise. $Z$ denotes white Gaussian
 of variance $1$ while $G$ denotes white Gaussian of variance
 $\sigma_G^2\geq 1$. 
\begin{eqnarray}
\nonumber&\expectp{\m{X}_0,\m{Z}}{J_2^{(\gamma)}(\m{X}_0,\m{Z})}\\
\nonumber&=  \int_{\m{z}}\int_{\m{x}_0}J_2^{(\gamma)}(\m{x}_0,\m{z}) f_0(\m{x}_0) f_Z(\m{z}) d\m{x}_0 d\m{z}\\
\nonumber &\geq  \int_{\m{z}\in\mathcal{S}_L^G}\left(\int_{\m{x}_0}J_2^{(\gamma)}(\m{x}_0,\m{z}) f_0(\m{x}_0) d\m{x}_0\right) f_Z(\m{z}) d\m{z}\\
&= \int_{\m{z}\in\mathcal{S}_L^G}\left(\int_{\m{x}_0}J_2^{(\gamma)}(\m{x}_0,\m{z}) f_0(\m{x}_0) d\m{x}_0\right)\nonumber\\
&\frac{f_Z(\m{z})}{f_G(\m{z})}f_G(\m{z}) d\m{z}.
\label{eq:beforeratio}
\end{eqnarray}
The ratio of the two probability density functions is given by
\begin{eqnarray*}
\frac{f_Z(\m{z})}{f_G(\m{z})}=\frac{e^{-\frac{\|\m{z}\|^2}{2}}}{\left(\sqrt{2\pi}\right)^m}\frac{\left(\sqrt{2\pi\sigma_G^2}\right)^m}{e^{-\frac{\|\m{z}\|^2}{2\sigma_G^2}}}=\sigma_G^m e^{-\frac{\|\m{z}\|^2}{2}\left(1-\frac{1}{\sigma_G^2}\right)}.
\end{eqnarray*}
Observe that $\m{z}\in\mathcal{S}_L^G$, $\|\m{z}\|^2\leq mL^2\sigma_G^2$. Using $\sigma_G^2\geq 1$, we obtain
\begin{equation}
\frac{f_Z(\m{z})}{f_G(\m{z})}
\geq \sigma_G^m 
e^{-\frac{m L^2 \sigma_G^2}{2}\left(1-\frac{1}{\sigma_G^2}\right)}
= \sigma_G^m e^{-\frac{mL^2(\sigma_G^2-1)}{2}}.
\label{eq:afterratio}
\end{equation}
Using~\eqref{eq:beforeratio} and~\eqref{eq:afterratio},
\begin{eqnarray}
&\nonumber\expectp{\m{X}_0,\m{Z}}{J_2^{(\gamma)}(\m{X}_0,\m{Z})} \\
\nonumber&\geq \sigma_G^m e^{-\frac{mL^2(\sigma_G^2-1)}{2}}\times \\
&\int_{\m{z}\in\mathcal{S}_L^G}\left(\int_{\m{x}_0}J_2^{(\gamma)}(\m{x}_0,\m{z}) f_0(\m{x}_0) d\m{x}_0\right)  f_G(\m{z}) d\m{z}\nonumber\\
\nonumber&=\sigma_G^m e^{-\frac{mL^2(\sigma_G^2-1)}{2}}\expectp{\m{X}_0,\m{Z}_G}{J_2^{(\gamma)}(\m{X}_0,\m{Z}_G)\indi{\m{Z}_G\in\mathcal{S}_L^G}}\\
&=\sigma_G^m
e^{-\frac{mL^2(\sigma_G^2-1)}{2}}\nonumber\\
&\expectp{\m{X}_0,\m{Z}_G}{J_2^{(\gamma)}(\m{X}_0,\m{Z}_G)|\m{Z}_G\in\mathcal{S}_L^G}\Pr(\m{Z}_G\in\mathcal{S}_L^G).
\label{eq:explb}
\end{eqnarray}
It is shown in~\cite{FiniteLengthsWitsenhausen} that
\begin{eqnarray}
\Pr(\m{Z}_G\in\mathcal{S}_L^G) = \frac{1}{c_m(L)}.
\label{eq:sphereprob}
\end{eqnarray}
 From~\eqref{eq:explb} and~\eqref{eq:sphereprob},
\begin{eqnarray}
\nonumber &\expectp{\m{X}_0,\m{Z}}{J_2^{(\gamma)}(\m{X}_0,\m{Z})}\\
\nonumber&\geq  \sigma_G^m
e^{-\frac{mL^2(\sigma_G^2-1)}{2}}\expectp{\m{X}_0,\m{Z}_G}{J_2^{(\gamma)}(\m{X}_0,\m{Z}_G)|\m{Z}_G\in\mathcal{S}_L^G}\nonumber\\
&(1-\psi(m,L\sqrt{m}))\nonumber\\
& =  \frac{\sigma_G^m e^{-\frac{mL^2(\sigma_G^2-1)}{2}}}{c_m(L)}  \expectp{\m{X}_0,\m{Z}_G}{J_2^{(\gamma)}(\m{X}_0,\m{Z}_G)|\m{Z}_G\in\mathcal{S}_L^G}.
\label{eq:ep0z}
\end{eqnarray}
We now need the following lemma, which connects the new finite-length lower bound to the length-independent lower bound of Theorem~\ref{thm:finiteratiosynergy}.
\begin{lemma}
\label{lem:epg}
\begin{eqnarray*}
&&\expectp{\m{X}_0,\m{Z}_G}{J_2^{(\gamma)}(\m{X}_0,\m{Z}_G)|\m{Z}_G\in \mathcal{S}_L^G}\\
&\geq &\left(\left(   \sqrt{\kappa_2 (P,\sigma_0^2,\sigma_G^2,L)} -\sqrt{P}           \right)^+\right)^2,
\end{eqnarray*}
for any $L>0$.
\end{lemma}

\begin{proof}
This is a reworking of the proof for the asymptotic case to a channel which has a truncated Gaussian noise of (pre-truncation) variance $\sigma_G^2$ and a truncation for $|Z_G|\leq L$. Details are omitted due to space constraints. The derivation follows exactly the lines of~\cite[Lemma 2]{FiniteLengthsWitsenhausen}.
\end{proof} 
The lower bound on the total average cost now follows from~\eqref{eq:ep0z} and Lemma~\ref{lem:epg}. 
\end{proof}

\bibliographystyle{IEEEtran}
\bibliography{IEEEabrv,MyMainBibliography}

\end{document}